\documentclass [journal,onecolumn,11pt]{IEEEtran}
\usepackage{etex}
\usepackage{amsfonts,amsmath,amssymb}
\usepackage{indentfirst, setspace}
\usepackage{url,float}
\usepackage{color}
\usepackage[a4paper,ignoreall]{geometry}
\usepackage{longtable}
\usepackage{graphicx}
\usepackage[a4paper,ignoreall]{geometry}
\usepackage{multicol}
\usepackage{stfloats}
\usepackage{enumerate}
\usepackage{cite}
\usepackage{amsthm}
\usepackage{multirow}
\usepackage{amssymb}
\usepackage{galois}
\usepackage[all]{xy}
\usepackage[square, comma, sort&compress, numbers]{natbib}
\usepackage{url,doi}
\usepackage{natbib,hyperref,doi}
\hypersetup{hidelinks}
\usepackage{comment}
\usepackage{verbatim}
\usepackage{tikz}
\usetikzlibrary{arrows}
\usepackage[justification=centering]{caption}

\def\qu#1 {\fbox {\footnote {\ }}\ \footnotetext { From Qu: {\color{red}#1}}}
\def\kq#1 {\fbox {\footnote {\ }}\ \footnotetext { From Quan: {\color{blue}#1}}}
\def\tl#1 {\fbox {\footnote {\ }}\ \footnotetext { From Niu: {\color{blue}#1}}}
\def\wang#1 {\fbox {\footnote {\ }}\ \footnotetext { From Wang: {\color{purple}#1}}}

\newcommand{\mkq}[1]{{{\color{blue}#1}}}

\newtheorem{Th}{Theorem}[section]
\newtheorem{Cor}[Th]{Corollary}

\newtheorem{Lem}[Th]{Lemma}

\newtheorem{example}[Th]{Example}

\newcommand{\tr}{{\rm Tr}}

\newcommand{\im}{{\rm Im}}
\newcommand{\gf}{{\mathbb F}}
\newcommand{\fqn}{\mathbb{F}_{q^n}}
\newcommand{\fq}{{\mathbb F}_q}

\makeatletter
\newcommand{\figcaption}{\def\@captype{figure}\caption}
\newcommand{\tabcaption}{\def\@captype{table}\caption}
\makeatother

\begin{document}

	\title{Finding compositional inverses of permutations from the AGW criterion}

		\author{
			{Tailin Niu, Kangquan Li, Longjiang Qu and Qiang Wang}
			\thanks{Tailin Niu, Kangquan Li and Longjiang Qu are with the College of Liberal Arts and Sciences, National University of Defense Technology, Changsha, 410073, China.
				Longjiang Qu is also with the State Key Laboratory of Cryptology, Beijing, 100878, China.
				Qiang Wang is with School of Mathematics and Statistics, Carleton University, 1125 Colonel By Drive, Ottawa, Ontario, K1S 5B6, Canada.
				
				The research of Longjiang Qu is partially supported by the National Key Research and Development Program of China under 2017YFB0802000,
				the Nature Science Foundation of China (NSFC) under Grant 61722213, 11771451, 62032009, and the Open Foundation of State Key Laboratory of Cryptology.
				The research of Qiang Wang is partially supported by NSERC of Canada.

				E-mail:
				Tailin Niu: runningniu@outlook.com,
				Kangquan Li: likangquan11@nudt.edu.cn,
				Longjiang Qu: ljqu\_happy@hotmail.com,
				Qiang Wang: wang@math.carleton.ca.	
			}
		}

	\maketitle{}
	\begin{abstract}		
		Permutation polynomials and their compositional inverses have wide applications in cryptography, coding theory, and combinatorial designs.
		Motivated by several previous results on finding compositional inverses of permutation polynomials of different forms, we propose a general method for finding these inverses of permutation polynomials constructed by the AGW criterion.
As a result, we have reduced the problem of finding the compositional inverse of such a permutation polynomial over a finite field to that of finding the inverse of a bijection over a smaller set.
	We demonstrate our method by interpreting several recent known results, as well as by providing new explicit results on more classes of permutation polynomials in different types.
		In addition, we give new criteria for these permutation polynomials being involutions.
		Explicit constructions are also provided for all involutory criteria.

		
	\end{abstract}
	
	
	\section{Introduction}
	Let  $q$ be a prime power and $\gf_q$ be the finite field with $q$ elements.
	We call a polynomial $f(x) \in\gf_q[x]$ a \textit{permutation polynomial} (PP for short) when the evaluation map $f: a \mapsto f(a) $ is a bijection.
	The unique polynomial $f^{-1}(x)$ over $\gf_q$ such that $ f \circ f^{-1} = f^{-1} \circ f =I $ is called \textit{the compositional inverse} of $f(x)$, where $ I $ denotes the identity map.
	Furthermore,  if  a PP $f$ satisfies $ f \circ f =I $, then $f$ is called an \textit{involution}.
	{Throughout this paper, for the multiplicative inverse function $f(x) = x^{-1}$, we always define $f(0) = 0$.
}
	Because PPs play important roles in finite field theory and  they have broad applications in coding theory,  combinatorial designs, and  cryptography \cite{
		carlet1998codes,ding2015permutation,hou2015determination,ding2006family,dobbertin1999almostw, 	
		muller1981some,lidl1984permutation,dobbertin1999almostn,mcfarland1973family,ding2013cyclic,dempwolff2013permutation}, the construction of PPs over finite fields has attracted a lot of attention.
	For recent surveys on constructing PPs, we invite the interested readers to consult  \cite{hou2015permutation}, \cite[Section 5]{li2018survey}, \cite{MullenWang14} and \cite[Section 5]{Wang2019index}.
	Explicitly determining the compositional inverse of a PP is useful because  both a PP and its inverse are required in many applications.
	For example, during the decryption process in a cryptographic algorithm with SPN structure,  the compositional inverse of the S-box plays an essential role.
	{Moreover, explicitly determining the compositional inverse will advance the further research of involutions, which are particularly useful (as part of a block cipher) in devices with limited resources \cite{charpin2016involutions,zheng2019constructions}.
	For recent research of involutions and  permutations with small cycles, we refer the readers to \cite{charpin2016involutions,zheng2019constructions,niu2019new,wuCharacterizationsConstructionsTriplecycle2020a,chenConstructionsCyclePermutations2020}.
}



	{ 	In general, it is difficult  to obtain the explicit compositional inverse of a random PP, except for several well known classes of PPs such as linear  polynomial,  monomials \cite{KyureghyanS12,KyureghyanS14}, and Dickson polynomials \cite{li2019compositional,lidlDicksonPolynomials1993}.
 In recent years, compositional inverses  of  several classes of PPs of special forms  have been obtained in explicit or implicit forms;  see  \cite{wu2013compositional, lidl1997finite,wuLinearizedPolynomialsFinite2013,wuCompositionalInverseClass2013a,tuxanidy2014inverses,wu2014compositional,li2019compositional,wang2017note,zheng2019constructions, zheng2019inverses,coulter2002compositional,tuxanidy2017compositional,zheng2018inverseoflarge}  for more details.
	A short survey on this topic can be found in \cite{zheng2019inverses}.
	In 2011, Akbary et al.  proposed a useful method called the AGW criterion for constructing PPs \cite{akbary2011constructing}.
		For the sake of convenience, when a PP is constructed using the AGW criterion or it can be interpreted by the AGW criterion, we call it an  \textit{AGW-PP}.
		Many classes of AGW-PPs  have been constructed up to  today and they can be divided into three types: multiplicative type \cite{gupta2016some,li2017several,li2017new,zha2017further,cepak2017permutations,li2018newp,tu2018two}, additive type  \cite{zheng2019two,akbary2011constructing,laigle-chapuyNoteClassQuadratic2007,yuan2011permutation} and hybrid type \cite{akbary2011constructing,kyureghyan2011constructing, li2018new, zheng2016large}.
Despite of recent progress on finding  compositional inverses of several classes of AGW-PPs, e.g.  \cite{wu2013compositional,tuxanidy2014inverses,li2019compositional,niu2019new},  there are  many  other classes of AGW-PPs whose compositional inverses are still unknown.  This motivates us to explore a general method to find  compositional inverses of these AGW-PPs.}

{
Inspired by the recent work in \cite{wu2014compositional,wu2013compositional,tuxanidy2014inverses,li2019compositional,niu2019new}, we propose a  general framework to solve the compositional inverse of an arbitrary AGW-PP, say $f$ on a finite set $A$.
It is well known that  there are  two surjective mappings $\lambda, \bar{\lambda}$ from $A$ to finite sets $S$ and $\bar{S}$, respectively, and a bijection $g$ from $S$ to $\bar{S}$ satisfying    $\overline{\lambda} \circ f = g \circ \lambda$ (i.e., Fig. 1 is commutative).
\begin{figure}[H]  \label{agwmap}
	\begin{equation*}
	\xymatrix{
		A \ar[rr]^{f}\ar[d]_{\lambda} &   &  A  \ar[d]^{\overline{\lambda}} \\
		S	 \ar[rr]^{g} &  & \overline{S} }
	\end{equation*}
	\caption{the AGW criterion} %
\end{figure}
In order to find  the inverse of the given AGW-PP $f$, we can construct two other mappings $ \eta, \overline{\eta} $ such that $\overline{\phi} \circ f = \psi \circ \phi$, where  $\phi=(\lambda, \eta)$
and  $ \overline{\phi}=(\overline{\lambda}, \overline{\eta})$ are bijections from $A$ to two other sets $\phi(A)$ and $\overline{\phi}(A)$ respectively (see Fig. 2).
	Namely, we have  the following commutative diagram and then the compositional inverse of $f$ is expressed by $f^{-1} =  \phi^{-1} \circ \psi^{-1} \circ \overline{\phi}$.
\begin{figure}[H] \label{agwinversemap}
	\begin{equation*}
		\xymatrix{
			A \ar[rr]^{f}\ar[d]_{\phi=(\lambda, \eta)} &   &  A  \ar[d]^{\overline{\phi}=(\overline{\lambda}, \overline{\eta})} \\
			\phi(A)	 \ar[rr]^{\psi=(g,\tau)} &  & \overline{\phi}(A) }
	\end{equation*}
\caption{a framework  to obtain the inverse} %
\end{figure}

Generally speaking, there are three types of AGW-PPs (multiplicative, additive, and hybrid) which are  classified based on  the properties of $ \lambda $ and $\overline{\lambda}$.
The classes studied by Li et al. \cite{li2019compositional} 
belong to the multiplicative case, while the classes studied by Wu \cite{wu2014compositional,wu2013compositional} and Tuxanidy et al. \cite{tuxanidy2014inverses,tuxanidy2017compositional} belong to the additive case.
Our general framework interprets all these recent results and provides a recipe to find compositional inverses of  many other classes of AGW-PPs.
	The key point of our approach is to design ``suitable" mappings $\eta$ and $\bar{\eta}$ such that  both $\phi =(\lambda, \eta)$  and $\bar{\phi} = (\bar{\lambda}, \bar{\eta})$ are bijections, and $\tau$ can be computed easily.
	Moreover,  compositional inverses of  $\phi$ and 	 $\psi=(g, \tau)$ can be efficiently computed.
	To demonstrate our approach, we propose several new explicit choices of $\eta$ and $\bar{\eta}$ and use them to find the compositional inverses of four classes of AGW-PPs in different types.
}
 {
As a consequence, we have reduced the problem of finding compositional inverse of a permutation polynomial over a finite field to that of finding the inverse of a bijection over a smaller set (see for example,  Theorems~\ref{mainmul}, Theorem~\ref{mainYD11}, Theorem~\ref{mainBk},  Theorem~\ref{mainGGF}.})

The rest of this paper is organized as follows.
In Section \ref{general}, we  present this unified method to find compositional inverses of any AGW-PP and recall some known results of computing compositional inverses under our framework.
As applications, we explicitly solve the compositional inverses of another four classes of AGW-PPs.
These results, as well as the characterization of involutions,   are divided into multiplicative, additive and hybrid cases, which are
presented in Sections \ref{multiplicative}, \ref{additive} and \ref{combinatorial} respectively.

\section{The unified method}
\label{general}
In this section, we present our unified approach to finding the compositional inverses of AGW-PPs.
First of all, we recall the following AGW criterion.
\begin{Lem}
	\label{LGWlemma}
	(\cite[Lemma 1.2]{akbary2011constructing}, the AGW Criterion)
	Let $A, S$, and $\overline{S}$ be finite sets with $\# S=\# \overline{S}$, and let $f: A\to A,$ $g: S\to \overline{S}$, $\lambda: A\to S$ and $\overline{\lambda}: A\to\overline{S}$ be maps such that $\bar{\lambda}\circ f=g\circ \lambda$.
	If both $\lambda$ and $\bar{\lambda}$ are surjective, then the following statements are equivalent:
	\begin{enumerate}[(1)]
		\item $f$ is a bijection and
		\item $g$ is a bijection from $S$ to $\overline{S}$ and $f$ is injective on $\lambda^{-1}(s)$ for each $s\in S$.
	\end{enumerate}
\end{Lem}
The AGW criterion can be illustrated in the commutative diagram of Fig. 1.
{This criterion  is very useful to explain many earlier constructions and to construct new classes of PPs.  The key of AGW criterion lies in transforming the problem of constructing permutations $f$ of  a finite set $A$ into finding a  bijections $g$ from $S$ to $\bar{S}$, whose cardinalities are both smaller than  the size of $A$.
Since $\lambda$ and $\overline{\lambda}$ both are surjective (normally they are not bijective because we prefer the sizes of $S$ and $\bar{S}$ are smaller), one can not simply obtain the compositional inverse of $f$ using only  mappings $\lambda$, $\overline{\lambda}$ and $g$.

When $A= \fqn$ and $\lambda= \bar{\lambda}$ are additive, Tuxanidy and Wang \cite{tuxanidy2014inverses}, \cite{tuxanidy2017compositional} converted the problem of computing the inverse of a PP over $\gf_q$  into that of computing two inverses of two other bijections over two  subspaces (one of them is  $\lambda\left(\mathbb{F}_{q^{n}}\right)$)  respectively.
A key ingredient is to decompose $\fqn$ into two subspaces.
This generalized a result  of  Wu et al.  \cite{wu2013compositional} who focused on the case that $\lambda$ is the trace function.
When  $A=\fq$ and  $\lambda= \bar{\lambda}$ are monomials,  Li et al. \cite{li2019compositional} provided  a multiplicative analogue of \cite{tuxanidy2014inverses, tuxanidy2017compositional, wu2013compositional}.
The main idea of  Li et al. relies on transforming the problem of computing the compositional inverses of permutation polynomials of the form $x^rh(x^s)$  over $\gf_q$ into computing the compositional inverses of two restricted permutation mappings, where one of them is a monomial over  $\mathbb{F}_q$ and the other is the polynomial $x^rh(x)^s$ over a particular subgroup of $\mathbb{F}_q$  with order $(q-1)/s$.
A key ingredient is a bijection from $\gf_q$ to another set $F_{q, s} =\{ (x^{q-s}, x^s) : x \in \fq\}$ whose inverse can be easily computed.
A variant of this result can be found in  Niu et al. \cite{niu2019new}.

We note that the similar idea works for an arbitrary AGW-PP $f$. Namely, we can find  the compositional inverse of  an arbitrary AGW-PP $f$ over a set $A$ by  constructing two other mappings, i.e., $ \eta $ and $  \overline{\eta} $ such that
	\begin{enumerate}[(1)]
\item $ \phi=(\lambda, \eta), \overline{\phi}=(\overline{\lambda}, \overline{\eta})$ become bijective from $A$ to some subsets of $A\times A$ as shown by Fig. 2.
\item Fig. 2 is a commutative diagram, i.e., $\overline{\phi} \circ f = \psi \circ \phi$.
	\end{enumerate}
}
{Then, it is clear that  the compositional inverse of $f$ can be expressed as $f^{-1} =  \phi^{-1} \circ \psi^{-1} \circ \overline{\phi}$.
Here we do not need to decompose the finite field into subspaces, as previously done for the additive case, neither we have restrictions
on special types of AGW-PPs.  Instead, we emphasize that it is crucial  to find simple mappings $\eta$ and $\overline{\eta}$ so that the compositional inverse of $\phi =(\lambda, \eta)$  and the compositional inverse of $\psi = (g, \tau)$ can be computed easily.  }

{To summarize the above discussion, we have the following theorem. }
\begin{Th}
	\label{main}
	Let $ A $ be a finite set, $f: A\to A$,  and let $\phi=(\lambda, \eta)$
	and $\overline{\phi}=(\overline{\lambda}, \overline{\eta})$
	be two bijective mappings from $A$ to some subsets of $A \times A$,
	and  denote by $ {\phi }^{-1}$, $ {\bar{\phi}}^{-1} $ their compositional inverses respectively.
	Let $\psi=(g,\tau): \phi(A)\to \overline{\phi}(A)$ be a mapping such that $\bar{\phi}\circ f=\psi\circ \phi$.
	Then $ f $ is bijective if and only if $ \psi $ is bijective.
	Furthermore, if $ \psi $ is bijective and its compositional inverse is denoted by $ \psi^{-1} $, then
	$$f^{-1}=  {\phi }^{-1} \circ \psi^{-1} \circ {\bar{\phi}}$$
	is the compositional inverse of $f$ on $A$.
\end{Th}
{We remark that the above theorem can be viewed as a special version of the AGW criterion, where the mappings $\lambda$ and $\bar{\lambda}$ are both bijections and the cardinalities of $S$ and $\bar{S}$ are both equal to that of $A$.
However, earlier constructions of AGW-PPs focused on bijections over smaller sets.
This result can be viewed as a new application of the AGW criterion in computing compositional inverses.

In fact, our method provides a possibility to solve the compositional inverses of all AGW-PPs and the process can be summarized as follows:
\begin{enumerate}[(1)]
	\item Design $\eta$ and $\overline{\eta}$ such that $\phi = (\lambda,\eta)$ and $\overline{\phi}=(\overline{\lambda},\overline{\eta})$ are both bijective, and compute the compositional inverse $\phi^{-1}$ of $\phi.$
	\item Compute the unique expression of $\psi$ such that $\psi\circ \phi =\overline{\phi}\circ f $.
	\item Compute the compositional inverses $\psi^{-1}$ of $\psi$.
	\item Obtain the compositional inverse $f^{-1}$ of $ f $ by Theorem \ref{main}, i.e., $ f^{-1}={\phi }^{-1} \circ \psi^{-1} \circ {\bar{\phi}} $.
\end{enumerate}

}

{In the following, we  demonstrate the explicit choices of $\phi, \bar{\phi}, \psi$  in several known results using  Theorem~\ref{main}.
In order to improve readability, we have adapted all the notations of these results in terms of  Lemma 2.1 and Theorem~\ref{main}. }

\begin{example}\cite[Theorem 2.3]{li2019compositional}\label{liExample}
	We take	
	\begin{enumerate}[(1)]
		\item $ A=\gf_q $;
		\item $f(x)=x^rh\left(x^s\right) $ permutes $ \gf_q  $, where $ s \mid (q-1), \gcd(r,q-1) =1 $ and $ h(0) \ne 0 $;
		\item $ \phi(x) =\overline{\phi}(x)= \left( x^{q-s}, x^s \right)$;
		\item $ \phi(A)=\overline{\phi}(A) $; 
		\item $ \psi(y,z)=\left(y^{r} h(z)^{q-s}, z^{r} h(z)^{s}\right) $; and
		\item Let $l(x)$ be the compositional inverse of $g(x)=x^rh(x)^s$ over $\mu_{\frac{q-1}{s}} =\left\{x\in{\gf}_{q}^*: x^\frac{q-1}{s}=1\right\}$  and $r'$ be an integer which satisfies $rr'\equiv1\pmod {q-1}$.
	\end{enumerate}
	Then it follows from Theorem \ref{main} that $$f^{-1}(x)=\left(\alpha(x) h(l(\beta(x)))^{s-1}\right)^{r^{'}}l(\beta(x))$$ is the compositional inverse of $ f(x) $, where $\alpha(x)=x^{q-s}$ and $\beta(x)=x^s$.
\end{example}

\begin{example}\cite[Theorem 3.7]{niu2019new}\label{niuExample}
	Assume that
	\begin{enumerate}[(1)]
		\item $ A=\gf_{q^m} $;
		\item $ f(x) = g\left(x^{q^i} - x + \delta\right) + cx \in \gf_{q^m}[x] $ permutes $ \gf_{q^m} $, where $q$ is a prime power, $m,i$ be positive integers with $ 1 \le i \le m-1 $, $ c \in \gf_{q^{\gcd(i,m)}}^* $ and $ g(x) \in \gf_{q^m}[x] $.  
		Then $ h(x) = g(x)^{q^i} - g(x) + cx +(1-c)\delta \in \gf_{q^m}[x]$ permutes $ \gf_{q^m} $ (see \cite[Proposition 3]{zheng2019two}), where $ \delta\in \gf_{q^m} $.
		\item $ \phi(x) =\overline{\phi}(x)=\left( -x^{q^i}, x^{q^i}-x+\delta \right) $;
		\item $  \phi(A)=\overline{\phi}(A) $;
		\item $ \psi(y,z)=\left(c^{q^i}y-g(z)^{q^i}, h(z) \right)$; and
		\item Assume $H(x)$ is the compositional inverse of $ h(x) $.
	\end{enumerate}
	Then it follows from Theorem \ref{main} that, for any $ \delta \in \gf_{q^m} $, the compositional inverse of $ f(x) $ is
	$$f^{-1}(x)= c^{-1}x^{q^i} - c^{-1}g\left(H(x^{q^i}-x+\delta)\right)^{q^i} -H(x^{q^i}-x+\delta)+\delta.$$
\end{example}

\begin{example}\cite[Theorem 1.2]{tuxanidy2014inverses}
	\label{tuxanidy2014inverses}
	Let
	\begin{enumerate}[(1)]
		\item $ A=\gf_{q^n} $;
		\item $f(x)=h(\psi_0(x)) \varphi_0(x)+g_0(\psi_0(x))$ permute {$\gf_{q^n}$}, where  $ {\varphi_0}, \psi_0 \in \gf_{q^n}[x] $ are additive polynomials, $ q $-polynomial $ \overline{\psi_0} $ satisfies $ \varphi_0 \circ \psi_0=\overline{\psi_0} \circ \varphi_0 $ and $ \left|\psi_0\left({\gf_{q^{n}}}\right)\right|=\left|\overline{\psi_0}\left( {\gf_{q^{n}}}\right)\right| $, and polynomial $ h \in \gf_{q^n}[x] $ such that $ h\left(\psi_0\left(\mathbb{F}_{q^{n}}\right)\right) \subseteq \mathbb{F}_{q} \backslash\{0\} $;
		\item $ \phi(x) =\phi_{\psi_0}\left(  x \right) =\left( \psi_0(x), x-\psi_0(x) \right)$, $ \overline{\phi}=\phi_{\overline{\psi_0}}\left( x   \right) = \left( \overline{\psi_0}(x), x-\overline{\psi_0}(x) \right) $  in \cite[Lemma 2.10]{tuxanidy2014inverses};
		\item $ \phi(A)=\phi_{\psi_0}\left(\mathbb{F}_{q^{n}}\right) ,\overline{\phi}(A)=\phi_{\overline{\psi_0}}\left(\mathbb{F}_{q^{n}}\right) $;
		\item $ \psi(y,z)=\left(h(y) \varphi_0(y)+\overline{\psi_0}(g_0(y)), h(y) \varphi_0(z)+g_0(y)-\overline{\psi_0}(g_0(y)) \right)$; and
		\item Assume that $ \left|S_{\psi_0}\right|=\left|S_{\overline{\psi_0}}\right| $ and $\ker(\varphi_0) \cap \psi_0\left(S_{\psi_0}\right)=\{0\} $.
		Then $ \varphi_0 $ induces a bijection from $ S_{\overline{\psi_0}} $ to $ S_{\overline{\psi_0}} $.
		Let $ \overline{f}^{-1}$ and $\left.\varphi_0^{-1}\right|_{S_{\overline{\psi_0}}} \in \gf_{q^{n}}[x] $ induce the inverses of $ \left.\overline{f}\right|_{\psi_0\left(\mathbb{F}_{q^n}\right)} (x)  =h(x) \varphi_0(x)+\overline{\psi_0}(g_0(x)) $ and $ \varphi_0 |_{S_{\psi_0}} $ respectively,
	\end{enumerate}
	Then it follows from Theorem \ref{main} that the compositional inverse of $ f(x) $ is given by
	$$ f^{-1}(x)=\overline{f}^{-1}(\overline{\psi_0}(x))+\left.\varphi_0^{-1}\right|_{S_{\overline{\psi_0}}}\left(\frac{x-\overline{\psi_0}(x)-g_0\left(\overline{f}^{-1}(\overline{\psi_0}(x))\right)+\overline{\psi_0}\left(g_0\left(\overline{f}^{-1}(\overline{\psi_0}(x))\right)\right)}{h\left(\overline{f}^{-1}(\overline{\psi_0}(x))\right)}\right) .$$
\end{example}

{
}

\begin{example}\cite[Theorem 2.3]{wu2013compositional} \label{wuExample}
	Let
	\begin{enumerate}[(1)]
		\item $ A=\gf_{q^n} $, where $ q $ is even and $ n $ is odd;
		\item $f(x)=x(L(\tr_{q^n/q} (x))+a \tr_{q^n/q} (x)+a x)$ permute $ \gf_{q^n} $, where $ xL(x) $ is a bilinear PP over $ \gf_q $ for a linearized polynomial $ L(x) \in \gf_q [x] $, $ a \in \gf_q^* $, and the trace function from $\gf_{q^n}$ to $\gf_{q}$ is denoted by $\tr_{q^n/q}(\cdot): x \to \sum_{i=0}^{n-1}x^{q^i}$;
		\item $ \phi(x) =\overline{\phi}(x)=\left( \tr_{q^n/q}(x), x + \tr_{q^n/q}(x) \right)$;
		\item $ \phi(A)=\overline{\phi}(A)=\gf_{q} \oplus  { \ker( \tr_{q^n/q})  }$;
		\item $ \psi(y,z)=\left(y L(y), a z^{2}+(L(y)+a y) z\right)$; and
		\item Let $ q = 2^m $ for a positive integer $ m $.
		Assume the compositional inverse of $ xL ( x ) $ is $ g_0( x ) \in \gf_q [ x ]  $.
	\end{enumerate}
	Then it follows from Theorem \ref{main} that the compositional inverse of $ f(x) $ is
	$$   f^{-1}(x)=a^{2^{m-1}-1} x^{2^{n m-1}}+\left(g_0(\tr_{q^n/q}(x))+a^{2^{m-1}-1} \sum_{k=1}^{\frac{n-1}{2}} x^{2^{(2 k-1) m-1}}\right)                                                         $$
	$$ {\times} \left(\frac{\tr_{q^n/q}(x)}{g_0(\tr_{q^n/q}(x))}+ag_0(\tr_{q^n/q}(x))\right)^{q-1}    $$
	$$  +\sum_{j=0}^{m-2} a^{2^{j}-1}\left(\frac{\tr_{q^n/q}(x)}{g_0(\tr_{q^n/q}(x))}+ag_0(\tr_{q^n/q}(x))\right)^{2^{m}-2^{j+1}}\left(\sum_{k=0}^{\frac{n-1}{2}} x^{q^{2 k}}\right)^{2^{j}}.  $$
\end{example}


{
As illustrated above, the key step of  this approach is to design suitable mappings  $\phi = (\lambda, \eta)$, $\bar{\phi} = (\overline{\lambda}, \overline{\eta})$ satisfying the required properties.   In the following we provide two new results which generalize the choices of $\eta$ and $\overline{\eta}$ in  Examples~\ref{liExample}-\ref{tuxanidy2014inverses}.
}


{

\begin{Cor} 	\label{mainc1}
	Let $A$ be a finite set and  $f$, $g$,  $\lambda, \overline{\lambda}$ be mappings satisfying the assumption of Lemma~\ref{LGWlemma} (e.g.,  satisfy the commutative diagram  in Fig.~1).  
	We assume $ \eta(x)=P(x)-\lambda(x) , \overline{\eta}(x)=P(x) - \overline{\lambda}(x)$, and $\tau$ are mappings   such that
	$\bar{\eta}\circ f=\tau \circ \eta$, where $ P(x) $ permutes $ A $.
	Then both 	$\phi=(\lambda, \eta)$
	and $\overline{\phi}=(\overline{\lambda}, \overline{\eta})$
	are bijective  mappings from $A$ to some subsets of $A \times A$  and $ {\phi }^{-1}(y,z)=P^{-1}(y+z)$, $ {\bar{\phi}}^{-1}(\alpha,\beta) =P^{-1}(\alpha+\beta) $ are their compositional inverses respectively.
Moreover, $ f $ is bijective  if and only if $ \psi =(g, \tau)$ is bijective.
	Furthermore, if both $ g $, $\psi $ are bijective and their compositional inverses are  denoted by $ g^{-1}(\alpha) $ and $ \psi^{-1}(\alpha,\beta)=\left(g^{-1}(\alpha) , M(\alpha,\beta) \right) $ respectively,  where $ M(\alpha,\beta): \im(\overline{\phi}) \to \im(\eta)  $, then
	$$f^{-1}(x)=  P^{-1} \left(g^{-1}(\overline{\lambda}(x)) + M(\overline{\lambda}(x), P(x) - \overline{\lambda}(x) ) \right)  $$
	is the compositional inverse of $f$ on $A$.
\end{Cor}
\begin{proof}
		{
		Assume $\phi(x) = \phi(x^\prime)$.  Then $\lambda(x) = \lambda(x^\prime)$ and
		$P(x) - \lambda(x)  =\eta(x) = \eta(x^\prime) =  P(x^\prime) - \lambda(x^\prime)$. Hence $P(x) = P(x^\prime)$. Because
		$P(x)$ permutes $A$, we must have $x = x^\prime$ and thus $\phi$ is a bijection.
Plug $y=\lambda(x), z=P(x)-\lambda(x)$ into $P^{-1}(y+z)$, one can obtain $P^{-1}(\lambda(x)+P(x)-\lambda(x) ) = x$.
Thus $ {\phi }^{-1}(y,z)=P^{-1}(y+z)$.
Similarly,  $\bar{\phi}$ is bijective and $ {\bar{\phi}}^{-1}(\alpha,\beta) =P^{-1}(\alpha+\beta) $.
The rest of proof follows from  Theorem \ref{main}.
}
\end{proof}
\begin{Cor}
	\label{mainc2}
	Let $A$ be a finite set and  $f$, $g$,  $\lambda, \overline{\lambda}$ be mappings satisfying the assumption of Lemma~\ref{LGWlemma} (e.g.,  satisfy the commutative diagram  in Fig.~1).  
Let	$ \lambda(x), \bar{\lambda}(x)  \ne 0 $ for any $ x \in A $.
	We assume $ \eta(x)=P\left( \frac{x}{\lambda(x)} \right)$, $\overline{\eta}(x)=P\left( \frac{x}{\overline{\lambda}(x)} \right) $,  and $\tau$ are mappings   such that
	$\bar{\eta}\circ f=\tau \circ \eta$, where $ P(x) $ permutes $ A $.
	Then both 	$\phi=(\lambda, \eta)$
	and $\overline{\phi}=(\overline{\lambda}, \overline{\eta})$
	are bijective  mappings from $A$ to some subsets of $A \times A$  and $ {\phi }^{-1}(y,z)=yP^{-1}(z)$, $ {\bar{\phi}}^{-1}(\alpha,\beta) =\alpha P^{-1}(\beta) $ are their compositional inverses respectively.
Moreover, $ f $ is bijective  if and only if $\psi =(g, \tau) $ is bijective.
	Furthermore, if both $ g $, $ \psi $ are bijective and their compositional inverses are respectively denoted by $ g^{-1}(\alpha) $, $ \psi^{-1}(\alpha,\beta)=\left(g^{-1}(\alpha) , M(\alpha,\beta) \right) $, where $ M(\alpha,\beta): \im(\overline{\phi}) \to \im(\eta)  $, then
	$$f^{-1}(x)= g^{-1} \left(  \overline{\lambda}(x) \right)  P^{-1} \left( M\left(  \overline{\lambda}(x),P\left( \frac{x}{\overline{\lambda}(x)} \right)   \right)  \right)  $$
	is the compositional inverse of $f$ on $A$.
\end{Cor}
\begin{proof}
	{
	Assume $\phi(x) = \phi(x^\prime)$.  Then $\lambda(x) = \lambda(x^\prime)$ and
		$P(\frac{x}{ \lambda(x)})  =\eta(x) = \eta(x^\prime) =  P(\frac{x^\prime }{\lambda(x^\prime)})$.
		Because		$P(x)$ permutes $A$,  the latter implies that
		$\frac{x}{ \lambda(x) }= \frac{x^\prime }{\lambda(x^\prime)}$. Since  $\lambda(x) = \lambda(x^\prime)$,
		we must have 	 		$x = x^\prime$ and thus $\phi$ is a bijection.
	Plug $y=\lambda(x), z=P\left( \frac{x}{\lambda(x)} \right)$ into $yP^{-1}(z)$, one can obtain $\lambda(x)P^{-1}(P\left( \frac{x}{\lambda(x)} \right))= x$.
	Thus $ {\phi }^{-1}(y,z)=yP^{-1}(z)$.
	Similarly,  $\bar{\phi}$ is bijective and  $ {\bar{\phi}}^{-1}(\alpha,\beta) =\alpha P^{-1}(\beta) $.
	The rest of proof follows from  Theorem \ref{main}.
}
\end{proof}

}

We demonstrate more specific choices of $\eta$ and $\overline{\eta}$ in the next few sections and explicitly compute the compositional invereses of four more classes of AGW-PPs. 	
{In addition, we analyze the results of compositional inverses to obtain conditions for  being involutions.
Furthermore, we provide at least one explicit involutory construction for each involutory criterion for the purpose
of demonstration, although we believe that it may be not hard to find more general involutory constructions.
In the sequel,  we need the following lemma on involutions over finite sets.}
\begin{Lem}
	\label{maininvolution}
	Let $A$ and $S$ be finite sets, and let $f:A \to A$, $g: S \to S$, $\lambda:A \to S$ be maps such that $\lambda$ is surjective and $\lambda \circ f = g \circ \lambda$.
	Assume $f$ is an involution on $A$.
	Then $g$ is an involution on $S$.
\end{Lem}
{
\begin{proof}
	We obtain $\lambda = \lambda \circ f \circ f = g \circ \lambda \circ f = g \circ g \circ \lambda$.
	Since $\lambda$ is surjective, $g$ is an involution on $S$ according to $\lambda = g \circ g \circ \lambda$.	
\end{proof}
This is a direct consequence of \cite[Proposition 2.2]{niu2019new}.}  From now on,  we always assume $ g $ is an involution whenever we consider  the involution  $f$.


\section{Compositional inverses of AGW-PPs in the multiplicative case}
\label{multiplicative}


{
In order to state our results we need the following terminology in \cite{akbary2009permutation}.
For any nonconstant monic polynomial $f(x) \in \mathbb{F}_{q}[x]$ of degree $\leqslant q-1$ with $f(0)=0,$ let $r$ be the vanishing order of $f(x)$ at zero and let $f_{1}(x):=f(x) / x^{r}$.
Let $\ell$ be the least divisor of $q-1$ with the property that there exists a polynomial $h(x)$  such that $f_{1}(x)=h\left(x^{(q-1) / \ell}\right) .$
So $f(x)$ can be written uniquely as $x^{r} h\left(x^{(q-1) / \ell}\right) $. The integer $\ell$ is called the index of $f$.  The
 AGW criterion is very useful to study  PPs of the form $x^rh(x^{(q-1)/\ell})$ such that  $\ell < q-1$.
More details can be found in \cite{akbary2009permutation,wang2013cyclotomy}.
}


The following result was discovered independently by several authors, and we want to point it out that  it is actually the multiplicative case of the AGW criterion.

\begin{Lem}
\cite[Theorem 2.3]{park2001permutation}
\cite[Theorem 1]{wang2007cyclotomic}
\cite[Lemma 2.1]{zieve2009some}
\label{Cri}
Let $q$ be a prime power and $f(x)=x^rh\left(x^s\right) \in\gf_q[x]$, where $s=\frac{q-1}{\ell}$ and $\ell $ is an integer.
Then $f(x)$ permutes $\gf_q$ if and only if
\begin{enumerate}[(1)]
	\item $\gcd\left(r,s\right)=1$ and
	\item $g(x)=x^rh(x)^s$ permutes $\mu_{\ell},$ where $\mu_{\ell}=\left\{x\in{\gf}_{q}^*  : ~~ x^{\ell}=1\right\}$.
\end{enumerate}
\end{Lem}
For any $s\mid {(q-1)}$,  by Lemma \ref{Cri}, the key point to determine the permutation property of $f(x)=x^rh\left(x^s\right)$ over $\gf_q$ is to consider whether $g(x)=x^rh(x)^s$ permutes $\mu_{\ell}$ or not.
It should be noted that $g(x)$ is always restricted over $\mu_{\ell}$.
Also we always assume that $h(x)\neq 0$ for any $x\in\mu_{\ell}$.
Otherwise, it is easy to see that $  f(x)=x^rh(x^s) $ can not permute $ \gf_{q} $.

\mkq{Using our method introduced in Section~\ref{general}, we can explicitly give the  compositional inverse of  a PP  $f(x) = x^rh(x^s)$   in terms of the compositional inverse of $g(x) = x^rh(x)^s$.  This
extends an earlier result by Li et al. in \cite{li2019compositional}, who transformed the problem of computing the compositional inverse of $ x^rh\left(x^s\right) $ into computing the compositional inverses of two restricted permutation mappings, where one of them is $g(x)=x^rh(x)^s$, and the other is $ x^{q-s} $. }
In terms of the language described in Theorem \ref{main}, Li et al. used $\eta = \overline{\eta} = x^{q-s}$ in order to assure that $\phi = \overline{\phi}= \left(x^s,x^{q-s} \right)$ is a bijection. However, their result only works when $\gcd\left(r,q-1\right)=1$. 	
In fact, from Lemma \ref{Cri},   the permutation property of $f $ only requires  that $g$ is bijective and  $\gcd(r,s)=1$.  Therefore we deal with this most general case and fill up the gap.

\begin{Th}
\label{mainmul}
Let $f(x)=x^rh\left(x^s\right)$ defined as in Lemma \ref{Cri} be a permutation over $\gf_q$ and $g^{-1}(x)$ be the compositional inverse of $g(x)=x^rh(x)^s$ over $\mu_{\frac{q-1}{s}}$. Suppose  $ a $ and $b $ are two integers  satisfying $ as+br=1 $.
Then the compositional inverse of $f(x)$ in $\gf_q[x]$ is given by
$$f^{-1}(x)=g^{-1}(x^s)^a x^b h\left(g^{-1}(x^s)\right)^{-b}.$$
\end{Th}


\begin{proof}
Let $\phi$ be a map defined by $$\phi: \gf_q \to \phi(\gf_q),$$ $$x\mapsto \left( \lambda(x), \eta (x) \right) = \left(x^{s},x^r\right),$$
where $\phi(\gf_q)=\left\{ \left( x^{s}, x^{r} \right):  ~~ x\in \gf_{q}  \right\} .$
Given one element $\left( y,z\right)\in \phi(\gf_q)$, there exists $x_0\in\gf_{q}$ such that $y=x_0^s$ and $z=x_0^r$.
Furthermore, it is clear that $y^az^b = x_0^{as+br}=x_0$.
Thus $\phi$ is a bijection and $ \phi^{-1}(y,z)=y^az^b $.

Next, we compute the expression of $\psi$ {from the relation $\psi \circ \phi(x) =\phi \circ f(x) $}.
From a simple computation, we get
\begin{equation}
	\label{mul1c}
	\phi \circ f(x) = \left(   x^{sr}h(x^s)^s ,  x^{r^2}h(x^s)^r  \right).
\end{equation}
Substituting $ x^s $ and $ x^r $ in Eq. (\ref{mul1c}) with $ y $ and $ z $ respectively, we obtain
$$\psi(y,z): \phi(\gf_q) \to \phi(\gf_q),$$
$$(y,z) \mapsto\left(y^rh(y)^s, z^rh(y)^r \right).$$

Now we compute the compositional inverse of $\psi$.
For $(y,z), (\alpha,\beta)\in \phi(\gf_q)$ with $ \psi(y,z)=(\alpha,\beta) $, we have
\begin{equation*}
	\left\{
	\begin{aligned}
		y^rh(y)^s &=& \alpha,  \\
		z^rh(y)^r &=& \beta. \\
	\end{aligned}
	\right.
\end{equation*}
Then we get $$y=g^{-1}(\alpha).$$
Moreover,
\begin{equation}
	z^r=\beta h(y)^{-r} = \beta h\left( g^{-1}(\alpha)\right)^{-r}.    \label{zzz}
\end{equation}
Clearly, for any given $(\alpha, \beta)\in\phi(\gf_q)$, there exists a unique element denoted by {${x_{(\alpha,\beta)}}\in\gf_q$} such that $ \alpha={x_{(\alpha,\beta)}}^s$  and $ \beta={x_{(\alpha,\beta)}}^r $.
Therefore, $z^r = \beta h\left( g^{-1}(\alpha)\right)^{-r} = {x_{(\alpha,\beta)}}^rh\left( g^{-1}({x_{(\alpha,\beta)}}^s)  \right)^{-r}$. To show that $z= x_{(\alpha,\beta)}h\left( g^{-1}({x_{(\alpha,\beta)}}^s)  \right)^{-1} $ is the unique solution to Eq. (\ref{zzz}), it suffices to prove that $x_{(\alpha,\beta)}h\left( g^{-1}({x_{(\alpha,\beta)}}^s)  \right)^{-1} \in \{x^r: ~~ x \in \gf_q\}$.
Since $f$ is a PP, there exists a unique $x_0\in\gf_q$ such that $ {x_{(\alpha,\beta)}}=f(x_0)=x_0^r h\left(x_0^s\right)  $.
Furthermore, we have $$g^{-1}(f(x_0)^s) =g^{-1}(x_0^{rs}h(x_0^s)^s) = g^{-1}(g(x_0^s)) = x_0^s. $$
Plugging ${x_{(\alpha,\beta)}}=x_0^r h\left(x_0^s\right)$ into $ {x_{(\alpha,\beta)}}h\left(  g^{-1}({x_{(\alpha,\beta)}}^s)  \right)^{-1} $, we obtain $${x_{(\alpha,\beta)}}h\left(  g^{-1}({x_{(\alpha,\beta)}}^s)  \right)^{-1} = x_0^{r} h\left(x_0^s\right) h(g^{-1}(f(x_0)^s))^{-1} = x_0^r h\left(x_0^s\right) h\left(x_0^s\right)^{-1} =x_0^r,$$ which belongs to $\{x^r: ~~ x \in \gf_q\}$.
Thus $ z={x_{(\alpha,\beta)}}h\left(  g^{-1}({x_{(\alpha,\beta)}}^s)  \right)^{-1} $ and then $$\psi^{-1}\left(\alpha,\beta \right)=\left( g^{-1}(\alpha), {x_{(\alpha,\beta)}}h\left(  g^{-1}({x_{(\alpha,\beta)}}^s)  \right)^{-1} \right).$$

{According to Theorem \ref{main},  the compositional inverse of $ f(x) $ is
$$f^{-1}(x)={\phi}^{-1} \circ \psi^{-1} \circ {\phi}(x)= g^{-1}(\alpha)^a  z^b  = g^{-1}(x^s)^a x^b h\left(g^{-1}(x^s)\right)^{-b}.$$}
\end{proof}

%
%
%
%
%
%
%
%
%
%

\mkq{
Theorem \ref{mainmul} can be verified by $f^{-1}(f(x))=x$ for any $x\in\gf_{ q}^*$.
First, it is clear that $ g^{-1}(f(x)^s) =g^{-1}(x^{rs}h(x^s)^s) = g^{-1}(g(x^s)) = x^s$.
Then we have
\begin{equation*}
\begin{aligned}
 f^{-1}\left(f(x)\right) &=  g^{-1}(f(x)^s)^a (x^rh\left(x^s\right))^b h\left(g^{-1}(f(x)^s)\right)^{-b}   \\
  &=x^{as} x^{br} h(x^s)^b h(x^s)^{-b} \\
  &=x^{as+br}=x.
 \end{aligned}
\end{equation*}
}
{Although this  provides a shorter proof of Theorem~\ref{mainmul}, we preferred to  give the current proof so that we can demonstrate how to use our approach to find the compositional inverse.}

Theorem \ref{mainmul} provides the explicit compositional inverse of a PP of the form $f(x) = x^rh(x^s)$ on $ \gf_q $ by computing the compositional inverse of $g(x) = x^rh(x)^s$ on $ \mu_{\frac{q-1}{s}} $, extending the result \cite[Theorem 2.3]{li2019compositional} which needs an additional condition $\gcd\left(r,q-1\right)=1$.
{ In the following, we explain that Theorem \ref{mainmul} is consistent with \cite[Theorem 2.3]{li2019compositional} under the  condition $\gcd\left(r,q-1\right)=1$.
In \cite[Theorem 2.3]{li2019compositional}, the authors obtained that the compositional inverse of $f(x) = x^rh(x^s)$ is
\begin{equation}
	\label{Li_in}
	f^{-1}(x)=\left( x^{q-s}h(g^{-1}(  x^s  ))^{s-1}\right)^{r'}g^{-1}(x^s),
\end{equation}
 where $\gcd\left(r,q-1\right)=1$ and $r'$ be an integer satisfying $rr'\equiv1\pmod {q-1}$.
Let $k$ be an integer satisfying $ r'  r+ k(q-1)=1$.
Since $\gcd\left(r,s\right)=1$, we assume $k_a s    +  k_b r=1$.
By letting $a=1   +k(1-s) k_a(q-1)$  and  $ b=  (1-s) r'  +k(1-s) k_b(q-1) $, one can verify $as+br=1$ and thus derive
 \eqref{Li_in} using the expression of $f^{-1}(x)$ in Theorem \ref{mainmul}. }
As consequences, many explicit compositional inverses of PPs of the form $ x^rh(x^s) $  given in \cite{li2019compositional}
can be obtained without the assumption $\gcd\left(r,q-1\right)=1$.

We remark that an explicit expression of the compositional inverse of $x^rh(x^s)$ in terms of roots of unities was given in  \cite{wang2009} and { more generally the inverses of cyclotomic mappings were provided in \cite[Theorem 2]{wang2017note} and \cite[Theorem 3.3]{zheng2016piecewise}.
In \cite{wang2017note}, a fast algorithm to generate cyclotomic PPs, their inverses, and involutions was provided.}
In contrast, our method explores the connections between the inverses of $f$ and $g$, sometimes, it can help us to obtain simpler expression of the compositional inverse.

{
Next we demonstrate  how to use Theorem \ref{mainmul} to obtain the compsitional inverses of  PPs in  \cite[Theorem 1.2]{zievePermutationPolynomialsForm2009} or \cite[Theorem 4.1]{akbaryPolynomialsFormRh2008}.
There are many concrete examples satisfying conditions in Corollary \ref{lalala}, and one of them can be found in Example \ref{kuozhancor}.
\begin{Cor}\label{lalala}
Let $s \mid(q-1), h(x) \in \mathbb{F}_{q}[x]$ satisfies that $h(\zeta)^{s}=\zeta^{n}$ for every $\zeta \in \mu_{(q-1) / s} $.
Suppose $\gcd(r+n,(q-1) / s)=1$ and $t  , a,b $ be integers that satisfy $(r+n) t \equiv 1(\bmod (q-1) / s)$ and $ar+bs=1$.
Then $f(x)=x^{r} h\left(x^{s}\right)$ permutes $\mathbb{F}_{q}$, $g(\zeta)=\zeta^{r} h(\zeta)^{s}=\zeta^{r+n}  $,  and $g^{-1}(\zeta)=\zeta^t$.
Moreover, the compositional inverse of $f(x)$ in $\mathbb{F}_{q}[x]$ is
$$f^{-1}(x)=x^{ast+b}  h\left(x^{st}\right)^{-b}.$$
\end{Cor}
}

In addition, we give a criterion for PPs of the form $x^rh(x^s) $ being involutions.
\begin{Cor}
\label{mulinvolution}
The PP $f(x) = x^rh(x^s) $ over $\gf_q$ defined as in Lemma \ref{Cri} is an involution if and only if
\begin{enumerate}[(1)]
	\item $g(x)=x^rh(x)^s$ is involutory on $\mu_{\frac{q-1}{s}}$ and
	\item $\varphi(x)=g(x^s)^a x^{b-r} h\left(g(x^s)\right)^{-b}h(x^s)^{-1}= 1 $ holds for any $ x \in \gf_q^* $, where integers $ a $ and $b $ satisfy $ as+br=1 $.
\end{enumerate}
\end{Cor}
\begin{proof}
Since $ g(x)$ being involutory is necessary for $ f(x) $ being involutory by Lemma \ref{maininvolution},  we have $g^{-1}(x) = g(x)$ on $\mu_{\frac{q-1}{s}}$.
In this case, by Theorem \ref{mainmul}, $f^{-1} (x) = f(x) $  if and only if $$ g(x^s)^a x^{b-r} h\left(g(x^s)\right)^{-b}h(x^s)^{-1}= 1.$$
\end{proof}

{
Many explicit classes of involutions in \cite{zheng2019constructions,niu2019new} can be constructed by Corollary \ref{mulinvolution}. We give the following example that has significantly simplified the earlier proof.  }
\begin{example}\cite[Corollary 2.17]{niu2019new}
	\label{kuozhancor}
Let $q$ be a power of $2$ and $ k $ be a positive integer such that $ \gcd(k,q+1)=1 $.
Then for any $\gamma , \beta  \in \gf_q^*$ such that $ \tr_{q/2}(\beta)=0 $, the polynomial $f(x)=x^{q^2-2}h\left(x^{q-1}\right)$ is an involution on $\gf_{q^2}$, where $h(x)=\gamma\left(  x^{-1}+\beta x^{-k-1}+\beta x^{k-1}\right) $.
\end{example}
\begin{proof}	
	According to the proof of \cite[Corollary 2.17]{niu2019new}, we obtain $h(x) \ne 0$ and $g(x)=x^{-1}h(x)^{q-1}=x$, for any $ x \in \mu_{q+1} $.
	{
Let $a= -1,   b= -q, r=-1, s=q-1$ in Corollary \ref{mulinvolution}.
	Then $\varphi(x)=x^{1-q} x^{-q+1} h\left(x^s\right)^{q}h(x^s)^{-1}=  x^{1-q} (x^{-q+1} h(x^{q-1})^{q-1}) = x^{1-q} x^{q-1} = 1$.
}	Thus $f$ is an involution.
\end{proof}

%

\section{Compositional inverses of AGW-PPs in the additive case}
\label{additive}


As demonstrated in Section~\ref{general}, our method can be used to find the compositional inverses for several classes of AGW-PPs in the additive case.
In this section, we further illustrate the new method by providing explicit compositional inverse of one more class of AGW-PPs of this type.
These PPs are of the form
$ g (x) + g_0(\lambda(x)) $ which is specified in Lemma \ref{YD11}.



\begin{Lem}
\cite[Theorem 6.1]{yuan2011permutation}
\label{YD11}
Assume that {$ F $} is a finite field and
$ S,  \overline{S}$ are finite subsets of $ F $ with $ \#S = \#\overline{S} $ such that
the maps $ \lambda : F \rightarrow S $ and $ \overline{\lambda}: F \rightarrow \overline{S} $ are surjective and $\overline{\lambda}$ is additive, i.e.,
$$\overline{\lambda}( x + y) = \overline{\lambda}(x) + \overline{\lambda}(y),\ \ \ \ \ x, y \in F.$$
Let $ g_0: S \rightarrow F $, and $ g: F \rightarrow F $  be maps such that $$\overline{\lambda} \circ ( g+g_0 \circ \lambda )= g  \circ \lambda , $$
{$g(S)=\bar{S}$} and $ \overline{\lambda}( g_0(\lambda(x))) = 0 $ for every $ x \in F $.
Then the map $ f(x) = g (x) + g_0(\lambda(x)) $ permutes $ F $ if and only if $ g $ permutes $ F $.
\end{Lem}

The  commutative diagram for the above AGW-PP is as follows.
\begin{equation*}
\xymatrix{
	F \ar[rr]^{f=g+g_0 \circ \lambda}\ar[d]_{\lambda} &   &  F  \ar[d]^{\overline{\lambda}} \\
	S	 \ar[rr]^{g} &  & \overline{S} }
\end{equation*}


{

The compositional inverse of the PP in Lemma \ref{YD11} can be found in the following theorem.}


\begin{Th}
\label{mainYD11}
{Let the symbols be defined as in Lemma \ref{YD11}.
	Let  $ f(x) = g (x) + g_0(\lambda(x)) $ be a  permutation over $ F $}  and $g^{-1}(x)$ be the compositional inverse of $g(x)$ {over $S$}.
Then the compositional inverse of $ f(x) $ is given by
$$f^{-1}(x)= g^{-1}\left(x-g_0(g^{-1}(\overline{\lambda}(x))) \right) .$$
\end{Th}

\begin{proof}
Let $ L(x) $ be a linearized PP over $F$ and $\phi,\overline{\phi}$ be maps defined by $$\phi: F \to \phi(F) $$ $$ x   \mapsto \left( \lambda(x),  L(x)-\lambda(x) \right) ,$$
and $$ \overline{\phi}: F \to \overline{\phi}(F) $$ $$x \mapsto \left( \overline{\lambda}(x),  L(x)-\overline{\lambda}(x) \right) .$$

Then $\phi$ and $ \overline{\phi} $ are bijections, and for $(y,z)\in \phi(F)$,  we have $\phi^{-1}(y,z)=L^{-1}(y+z)$, 
where {$\phi^{-1}$} and $ L^{-1} $ denote the compositional inverses of {$\phi$} and $ L $ respectively.
Here is the commutative diagram as stated in Theorem \ref{main}.

\begin{equation*}
	\xymatrix{
		F \ar[rr]^{f=g + g_0\circ\lambda }\ar[d]_{\phi=(\lambda, \eta)} &   &  F  \ar[d]^{\overline{\phi}=(\overline{\lambda}, \overline{\eta})} \\
		\phi(F)	 \ar[rr]^{\psi=(g,\tau)} &  & \overline{\phi}(F) }
\end{equation*}

Next,   we find the map $\psi$ such that  $ \psi \circ \phi =\overline{\phi} \circ f  $.
After direct calculation, we have
\begin{equation}
	\label{additive1c}
	\overline{\phi}\circ f(x)=\left( g(\lambda(x)),    L \left(  g(x) \right) + L\left(  g_0(\lambda(x)) \right) -g(\lambda(x)) \right).
\end{equation}
To compute $ \psi(y,z) $ accordingly such that $\bar{\phi}\circ f(x)=\psi\circ \phi (x)$, we substitute $ \lambda(x) $ and $ L(x)-\lambda(x) $ in Eq. (\ref{additive1c}) with $ y $ and $ z $ respectively.
We obtain
$$\psi(y,z): \phi(F) \to \overline{\phi}(F),$$
$$(y,z) \mapsto\left( g(y), L(g(L^{-1}(y+z)))+ L( g_0(y) ) - g(y) \right).$$	

Clearly $\psi $ is a bijection, since $\psi \circ \phi(x) =\phi \circ f(x) $.
Now we compute the compositional inverse of $\psi$.
Since $ f $ permutes $ F $, $ g $ permutes $ F $, and we assume that $g^{-1}$ is the compositional inverse of $g$.
For any $(y,z)\in\phi(F)$ with $ \psi(y,z)=(\alpha,\beta) \in \overline{\phi}(F) $, we have
\begin{equation*}
	\left\{
	\begin{aligned}
		g(y)&=& \alpha,  \\
		L(g(L^{-1}(y+z)))+ L( g_0(y) ) - g(y)&=& \beta.  \\
	\end{aligned}
	\right.
\end{equation*}
Clearly, $$y=g^{-1}(\alpha).$$
Moreover,
\begin{equation}
	\label{YD112eq}
	L(g(L^{-1}(y+z)))=\alpha+ \beta-L( g_0(y) ).
\end{equation}
Composing $ L \circ g^{-1} \circ L^{-1} $ on Eq. (\ref{YD112eq}) and simplifying it, we have
\begin{equation*}
	z=L \left(  g^{-1} \left(        L^{-1}  ( \alpha+ \beta  )-  g_0(g^{-1}(\alpha)) \right)    \right)   - g^{-1}(\alpha).
\end{equation*}
Hence we have	$$\psi^{-1}\left(\alpha,\beta \right)=\left( g^{-1}(\alpha) ,   L \left(  g^{-1} \left(        L^{-1}  ( \alpha+ \beta  )-  g_0(g^{-1}(\alpha)) \right)    \right)   - g^{-1}(\alpha)        \right) .$$

Finally we compute the compositional inverse of $ f(x) $.
According to Theorem \ref{main},  together with 
$\alpha = \overline{\lambda}(x)$ and $\beta=L(x) -\overline{\lambda}(x)$, we obtain the compositional inverse of $ f(x) $ is
$$f^{-1}(x)={\phi}^{-1} \circ \psi^{-1} \circ \overline{\phi}(x)=  g^{-1} \left(        L^{-1}  ( \alpha+ \beta  )-  g_0(g^{-1}(\alpha)) \right)  = g^{-1}\left(x-g_0(g^{-1}(\overline{\lambda}(x))) \right) .$$
\end{proof}


{
	Theorem \ref{mainYD11} can be checked directly by $f^{-1}(f(x)) = x $ for any $x\in F$ (we omit the details here).
The compositional inverse of $ f(x)=g_1(\lambda(x))^{s}+g(x) $ in \cite[Corollary 6.2]{yuan2011permutation} is given in Corollary \ref{YD11cor}, as an example of Theorem \ref{mainYD11}.
Since  $g$ is a linearized PP on $\mathbb{F}_{q^{n}}$, its compositional inverse $g^{-1}$ can be explicitly obtained (see \cite{FF}).
}

{
\begin{Cor}
	\label{YD11cor}
Let $n$ and $k$ be positive integers such that $\gcd(n, k)=d>1$, let $s$ be any positive integer with $s\left(q^{k}-1\right) \equiv 0 \pmod {q^{n}-1} .$
Let $g(x) \in \mathbb{F}_{q}[x]$ be a linearized polynomial permuting $\mathbb{F}_{q^{n}}$, $\lambda(x)=\overline{\lambda}(x)$ be a $ q^{d}$-polynomial with $\lambda(1)=0$ and $g_1(x) \in \mathbb{F}_{q^{n}}[x] $.
Then the compositional inverse of $ f(x)=\left(g_1(\lambda(x))\right)^s+g(x) $ over $\mathbb{F}_{q^{n}}$ is
$$f^{-1}(x)= g^{-1}\left(x-g_1(g^{-1}({\lambda}(x)))^s \right) .$$
\end{Cor}
}


In addition, we give a criterion for PPs of the form $g (x) + g_0(\lambda(x))$ being involutions.
\begin{Cor}
\label{addinvolution1}
{In Lemma \ref{YD11}, let $\lambda=\bar{\lambda}$, $ S=\bar{S}$ be a finite subset of $ F $ such that $\lambda(F)=S$, and let $g$ be a bijection from $S$ to $S$.}
Then, the PP $ f(x) = g (x) + g_0(\lambda(x)) $ over $F$ defined as in Lemma \ref{YD11} is an involution if and only if
{$\varphi(x)=g^{-1}\left(x-g_0(g({\lambda}(x))) \right) - g (x) -g_0(\lambda(x)) = 0 $ holds for any $x\in F$}.
\end{Cor}
\begin{proof}
{Its proof can be easily derived by applying Theorem \ref{mainYD11} and thus we omit  the details.}
\end{proof}

Below, we provide some explicit constructions of involutions from Corollary \ref{addinvolution1}.
\begin{Cor}
	Let $ q  $ be an even prime power and $ \lambda( g_0(\lambda(x))) = 0 $ hold for any $ x \in F=\gf_{ q^n } $ defined as in Lemma \ref{YD11}, with $ \overline{\lambda}  = \lambda  $ being additive.
	Then the PP $ f  (x)= x + g_0(\lambda(x)) $ is an involution on $ \gf_{ q^n } $.
\end{Cor}

\begin{example}
	Let $ q  $ be a power of $ 2 $, $ S=\gf_q $, and $ g_0 $ be any polynomial such that $ g_0(\gf_q) \subseteq \gf_q $.
	Assume $ n $ is an even integer, $ \lambda(x)= \tr_{q^n/q}(x) $.
	Thus we have $ \tr_{q^n/q}(1)=0 $ and $ \tr_{q^n/q}( g_0(\tr_{q^n/q}(x))) =\tr_{q^n/q}( 1 )   g_0(\tr_{q^n/q}(x))= 0 $ holds for any $ x \in F=\gf_{ q^n } $.
	Then $ f  (x)= x + g_0(\tr_{q^n/q}(x)) $ is an involution on $ \gf_{ q^n } $.
\end{example}

\section{Compositional inverses of AGW-PPs in the hybrid case}
\label{combinatorial}

In this section, we use our method to study AGW-PPs in the hybrid case, and to obtain the compositional inverses of $ xh(\lambda(x)) $ in Lemma \ref{2-6.3} and $  x + \gamma G(\lambda(x)) $ in Lemma \ref{2-6.4} as examples.
{For each class, we obtain the inverses of  these PPs and  present some explicit  classes. We also provide an involutory criterion and demonstrate some involutory constructions.}

\begin{Lem}
\cite[Theorem 6.3]{akbary2011constructing}
\label{2-6.3}
Let $ q $ be any power of the prime number $ p $, let $ n $ be any positive integer, and let $ S $ be any subset of $ \gf_{q^n} $ containing 0.
Let $ h,k \in \gf_{q^n}$ be any polynomials such that $ h(0) \ne 0 $ and $ k(0) = 0 $, and let $ {\lambda(x)} \in \gf_{q^n}[x]  $ be any polynomial satisfying
\begin{enumerate}[(1)]
\item $h(\lambda(\gf_{q^n})) \subseteq S$; and
\item $ \lambda(a\alpha) = k(a) \lambda(\alpha) $ for all $ a \in S $ and all $ \alpha \in \gf_{q^n}$.
\end{enumerate}
Then the polynomial $ f (x) = xh(\lambda(x)) $ is a permutation polynomial for $ \gf_{q^n} $ if and only if $ g(x) = xk(h(x)) $ induces a permutation of $ \lambda(\gf_{q^n}) $.
\end{Lem}
The above AGW-PPs can be illustrated by the following commutative diagram.
\begin{equation*}
\xymatrix{
\gf_{q^n} \ar[rr]^{f(x)=xh(\lambda(x))}\ar[d]_{\lambda} &   &  \gf_{q^n}  \ar[d]^{\lambda} \\
\lambda(\gf_{q^n})	 \ar[rr]^{{g(x)} = xk(h(x))} &  & \lambda(\gf_{q^n}).}
\end{equation*}

Using our unified approach, we choose $\eta(x) = x-\lambda(x)$ such that $\phi(x) = (\lambda(x), x-\lambda(x))$ is bijective and thus obtain the compositional inverses in the following theorem.

\begin{Th}
\label{mainBk}
Let  the symbols be defined as in Lemma \ref{2-6.3}.
Let $ f(x) = xh(\lambda(x)) $  permute $ \gf_{q^n} $ and $g^{-1}(x)$ be the compositional inverse of $g(x)=xk(h(x))$ over $ \lambda(\gf_{q^n}) $.
Then the compositional inverse of $ f (x) $ is given by
$$f^{-1}(x)= \frac{x-\lambda(x) + k\left(h\left(g^{-1}\left(\lambda(x)\right)\right)\right)g^{-1}\left(\lambda(x)\right)}{h\left(g^{-1}\left(\lambda(x)\right)\right)} .$$
\end{Th}
\begin{proof}
Let  $\eta(x) = x-\lambda(x)$ and $\phi$ be a map defined by $$\phi: \gf_{q^n} \to \phi(\gf_{q^n}),$$ $$x \mapsto \left( \lambda(x),  x-\lambda(x) \right) .$$
Then $\phi$ is a bijection and for $(y,z)\in \phi(\gf_{q^n}),\ \phi^{-1}(y,z)=y+z.$
Let us consider the following commutative diagram.
\begin{equation*}
\xymatrix{
\gf_{q^n} \ar[rr]^{f(x)=xh(\lambda(x))}\ar[d]_{\phi=(\lambda, \eta)} &   &  \gf_{q^n}  \ar[d]^{\phi=(\lambda, \eta)} \\
\phi(\gf_{q^n})	 \ar[rr]^{\psi=(g,\tau)} &  & \phi(\gf_{q^n}) }
\end{equation*}

First we determine the expression of $\psi$ to {establish the connection} ${\phi}\circ f(x)=\psi\circ \phi (x)$. After direct computation, we have
\begin{equation}
\label{combi1c}
\phi \circ f(x)=\left(\lambda(x)k(h(\lambda(x))), xh(\lambda(x))-k(h(\lambda(x)))\lambda(x) \right) .
\end{equation}
Then, substituting $ \lambda(x) $ and $ x-\lambda(x) $ in Eq. (\ref{combi1c}) with $ y $ and $ z $ respectively, we obtain
$$\psi(y,z): \phi(\gf_{q^n}) \to \phi(\gf_{q^n}),$$
$$(y,z) \mapsto\left(yk(h(y)), (y+z)h(y)-k(h(y))y \right).$$	

Since $f(x)=xh(\lambda(x))$ permutes $\gf_{q^n}$, we have that both $\psi$ and $g(x) = xk(h(x))$ are bijective.
Recall that $ g^{-1} $ denotes the compositional inverse of $ g(y) = yk(h(y)) $ over $\lambda(\gf_{q^n})$.
In the following, we compute the compositional inverse of $\psi$.
Let $(y,z), (\alpha,\beta)\in \phi(\gf_{q^n}),$ satisfy $ \psi(y,z)=(\alpha,\beta) $, i.e.,
\begin{equation}
\left\{
\begin{aligned}
yk(h(y))&=& \alpha, \label{Bk2y} \\
(y+z)h(y)-k(h(y))y&=& \beta.  \\
\end{aligned}
\right.
\end{equation}
Then, we have $$y=g^{-1}(\alpha).$$
Moreover, for $ z $, we firstly explain $ h(y) \ne 0 $ for any $ y \in \lambda(\gf_{q^n}) $.
Assume that there exists some $ y_0 \in \lambda(\gf_{q^n}) $ such that $ h(y_0) = 0 $.
We have $ g(y_0) = y_0k(h(y_0)) =0$ and $ g(0) =0 $.
Since $ g(y) $ is bijective, we have $ y_0 =0 $, which is conflict with $ h(0) \ne 0 $.
Thus $ h(y) \ne 0 $ for any $ y \in \lambda(\gf_{q^n}) $.
Then it follows from Eq.(\ref{Bk2y}) that
\begin{equation*}
\label{Bk2eq}
z= \frac{\beta + k\left(h\left( y \right)\right) y}{h\left(y  \right)}-y.
\end{equation*}
Hence,  we have
$$\psi^{-1}\left(\alpha,\beta \right)=\left( g^{-1}(\alpha) ,   \frac{\beta + k\left(h\left(g^{-1}(\alpha)\right)\right)g^{-1}(\alpha)}{h\left(g^{-1}(\alpha)\right)}-g^{-1}(\alpha) \right).$$

Finally, we compute the compositional inverse of $ f(x) $.
From Theorem \ref{main}, together with the compositional inverses of $\phi, \psi$ and $\alpha = \lambda(x)$, $\beta=x-\lambda(x)$, the compositional inverse of $ f(x) $ is

\begin{equation*}
\begin{aligned}
f^{-1}(x)&= {\phi}^{-1} \circ \psi^{-1} \circ {\phi}(x)     \\
&=   g^{-1}(\alpha) + \frac{\beta + k\left(h\left(g^{-1}(\alpha)\right)\right)g^{-1}(\alpha)}{h\left(g^{-1}(\alpha)\right)} - g^{-1}(\alpha)   \\
&=  \frac{x-\lambda(x) + k\left(h\left(g^{-1}\left(\lambda(x)\right)\right)\right)g^{-1}\left(\lambda(x)\right)}{h\left(g^{-1}\left(\lambda(x)\right)\right)}.
\end{aligned}
\end{equation*}
\end{proof}


{Theorem \ref{mainBk} can also be verified by $f^{-1}(f(x)) = x$ for any $x\in F$ (we omit the details here).}
A special case of Theorem \ref{mainBk} with $ k(x) = x^2 $ and $ S = \gf_p $ (see \cite[Proposition 12]{marcos2011specific}) is given in the following corollary.
\begin{Cor}
\label{mainBkc}
Let $ \lambda(x) $ be either $ \lambda_2(x) =\sum\limits_{0\le i < j\le n - 1} {{x^{{p^i} + {p^j}}}}  $ or $ T_2(x)=\tr_{p^n/p}(x^2) $.
Let $ h(x)\in \gf_p[x] $ such that $h(0) \ne 0$.
If the polynomial $ g(x)=x(h(x))^2 $ permutes $\gf_{p} $, then the polynomial $ f(x)=xh(\lambda(x)) $ permutes $ \gf_{p^n} $, and the compositional inverse of $ f (x) $ is given by
$$f^{-1}(x)= \frac{x-\lambda(x) + \left(h\left(g^{-1}\left(\lambda(x)\right)\right)\right)^2  g^{-1}\left(\lambda(x)\right)}{h\left(g^{-1}\left(\lambda(x)\right)\right)} .$$
\end{Cor}

In addition, we propose a criterion for PPs of the form $xh(\lambda(x)) $ being involutions,  and give some involutory constructions.
\begin{Cor}
\label{cominvolution1}
Let $f(x)=xh(\lambda(x))$, $ g(x) = xk(h(x)) $ and $h, k, \lambda$ defined as in Lemma \ref{2-6.3}.
Let $\theta(x)=k(h(x))$ for any  $ x \in \gf_{q^n}  $.
Then $f$ is an involution over $ \gf_{q^n} $ if and only if
\begin{enumerate}[(1)]
\item $  \theta \left(  \theta(y)   y    \right)   \theta(y) =1 $ holds for any $ y \in \lambda(\gf_{q^n}^*) $, and
\item $\varphi(y)=h\left(   g(y)   \right)   h(y) -1 =0  $ holds for any {$ y \in \lambda(\gf_{q^n}^*) $}.
\end{enumerate}
\end{Cor}
\begin{proof}
Since $ f(0)=0 $, we only need to consider the nonzero situation.
{Assume $f$ is an involution. By Lemma \ref{maininvolution}}, we have $g$ is an involution, which is equivalent to $ \theta \left(  \theta(y)   y    \right)   \theta(y)   = k\left(h\left(  k(h(y))   y    \right)\right)   k(h(y)) =1 $ for $ y \in \lambda(\gf_{q^n}^*) $.
{From now on, we assume (1) hold and    prove that $f$ is an involution if and only if  $\varphi(x)=0  $.}
Plugging {$$ k\left(h\left(  k(h(y))   y    \right)\right)   k(h(y)) =1 $$ for $ y \in \lambda(\gf_{q^n}^*) $}, $ g^{-1}(x)= g(x) $ and $ f^{-1}(x)  $ by {Theorem \ref{mainBk}} into $ \varphi(x) $, we have
\begin{eqnarray*}
x\varphi(x)&=&  h\left(  k(h(\lambda(x)))   \lambda(x) \right)   xh(\lambda(x)) - x \\
&=&  h\left(  k(h(\lambda(x)))   \lambda(x) \right)   xh(\lambda(x)) - x +\lambda(x) - \lambda(x) k\left(h\left(  k(h(\lambda(x)))   \lambda(x)    \right)\right)   k(h(\lambda(x)))         \\	
&=&  h\left(g\left(\lambda(x)\right)\right)    \left( xh(\lambda(x)) -\frac{x-\lambda(x) + k\left(h\left(g^{-1}\left(\lambda(x)\right)\right)\right)g^{-1}\left(\lambda(x)\right)}{h\left(g^{-1}\left(\lambda(x)\right)\right)}   \right)  \\
&=& h\left(g\left(\lambda(x)\right)\right)  \left( f(x)-f^{-1}(x)    \right),
\end{eqnarray*}
for $ x \in  \gf_{q^n}^* $.
{Note $h(y) \ne 0$ because  $f(x)=xh(\lambda(x))$ permutes $\gf_{q^n}$ and $f(0)=0$. Hence $ f(x) $ is an involution if and only if $\varphi=0$.}
\end{proof}

\begin{example}
	\label{mainBkc1}
	Let $ q=3^n$.
	Then, $f(x) = x (\lambda(x)^2 + 1)$ is an involution over $\gf_q$, where $ \lambda (x)=\sum\limits_{0\le i < j\le n - 1} {{x^{{3^i} + {3^j}}}}$.
	\begin{proof}
		Let $h(x) = x^2  + 1$,  $ S= \gf_3$	in Lemma \ref{2-6.3}.
		We have $k(x) = x^2$ and $\lambda(\gf_q^*)=\gf_3$.
		One can obtain $g(y)=y(y^2  + 1)^2=y$ and $\theta(y)= (y^2+1)^2=1$  for $y \in \gf_3$.
		Hence $ \theta( \theta(y)y) \theta(y)=  \theta(  g(y)  ) \theta(y)    = \theta(y)^2  =1 $ and $\varphi(y)=h\left(   g(y)   \right)   h(y) -1=  h(y)^2-1=0 $.
		By Corollary \ref{cominvolution1},  $f(x)$ is an involution.
	\end{proof}
\end{example}


	%

%
%

In the following, we will give explicitly the compositional inverse of a PP containing a $b$-linear translator.
For $ S \subset \gf_q , \gamma, b \in \gf_q$ and a map $ {\lambda} : \gf_q \rightarrow \gf_q , \ \gamma $ is called a \textit{$b$-linear translator} \cite{akbary2011constructing,charpin2009does,kyureghyan2011constructing} of $ \lambda $ with respect to $ S $ if $ \lambda(x +u \gamma ) = \lambda(x) +ub $ for all $ x \in \gf_q $ and $ u \in S $.

\begin{Lem}
\cite[Theorem 6.4]{akbary2011constructing}
\label{2-6.4}
Let $ S \subseteq \gf_q $ and $ \lambda : \gf_q \rightarrow S $ be a surjective map.
Let $ \gamma \in \gf_q^* $ be a $b$-linear translator with respect to $S$ for the map $\lambda$.
Then for any $ G \in \gf_q[x] $ which maps $ S $ into $ S $, we have that $ f(x)= x + \gamma G(\lambda(x)) $ is a permutation polynomial of $ \gf_q $ if and only if $ g(x)= x + bG(x) $ permutes $S$.
\end{Lem}
It can be illustrated by the following commutative diagram.
\begin{equation*}
\xymatrix{
\gf_{q} \ar[rr]^{f}\ar[d]_{\lambda} &   &  \gf_{q}  \ar[d]^{\lambda} \\
\lambda(\gf_{q})	 \ar[rr]^{g} &  & \lambda(\gf_{q}) .}
\end{equation*}


Below, we describe how to obtaining the compositional inverses of PPs $  f (x) =  x + \gamma G(\lambda(x)) $.
The fundamental idea of our approach is the following commutative diagram, where we design $\eta(x) = x-\lambda(x)$ such that $\phi(x) = (\lambda(x), x-\lambda(x))$ is bijective, see Theorem \ref{mainGGF}.
\begin{equation*}
\xymatrix{
\gf_{q} \ar[rr]^{f}\ar[d]_{\phi=(\lambda, \eta)} &   &  \gf_{q}  \ar[d]^{\phi=(\lambda, \eta)} \\
\phi(\gf_{q})	 \ar[rr]^{\psi=(g,\tau)} &  & \phi(\gf_{q}) .}
\end{equation*}

We obtain the compositional inverses in the following theorem:
\begin{Th}
\label{mainGGF}
Let $ f(x)= x + \gamma G(\lambda(x)) $ defined as in Lemma \ref{2-6.4} be a PP on $ \gf_q $ and $g^{-1}(x)$ be the compositional inverse of $g(x)=x+bG(x)$.
Then the compositional inverse of $ f (x) $ is given by
$$f^{-1}(x)=(b-\gamma)G\left(g^{-1}(\lambda(x))\right)+g^{-1}(\lambda(x))  - \lambda(x) +x .$$
\end{Th}
\begin{proof}
Using the same notation and assumptions of Lemma \ref{2-6.4}, let $\phi$ be a map defined by $$\phi: \gf_{q^n} \to \phi(\gf_{q^n}),$$ $$x \mapsto \left( \lambda(x),  x-\lambda(x) \right) .$$
Then $\psi$ is a bijection and for $(y,z)\in \phi(\gf_{q^n}),\ \psi^{-1}(y,z)=y+z.$

Next, we determine the expression of $\psi$ {from the equation} ${\phi}\circ f(x)=\psi\circ \phi (x)$.
After direct computation, we have
\begin{equation}
\label{combi2c}
{\phi}\circ f(x)=\left( \lambda(x) + bG(\lambda(x)) , x-\lambda(x)+ (\gamma-b)G(\lambda(x)) \right).
\end{equation}
Substituting $ \lambda(x) $ and $ x-\lambda(x) $ in Eq. (\ref{combi2c}) with $ y $ and $ z $ respectively, we obtain
$$\psi(y,z): \phi(\gf_{q^n}) \to \phi(\gf_{q^n}),$$
$$(y,z) \mapsto\left( y + bG(y) , z+ (\gamma-b)G(y) \right).$$	

Since $f(x)=x + \gamma G(\lambda(x))$ is a PP on $\gf_{q}$, $\psi$ and $ g(x)= x + bG(x) $ are both bijective.
We assume that $ g^{-1} $ denotes the compositional inverse of $ g(x)= x + bG(x)  $ over $S$.
In the following, we compute the compositional inverse of $\psi$.
Let $(y,z), (\alpha,\beta)\in \phi(\gf_{q^n}),$ satisfy $ \psi(y,z)=(\alpha,\beta) $, i.e.,
\begin{equation*}
\left\{
\begin{aligned}
y + bG(y)&=& \alpha, \label{GGFy} \\
z+ (\gamma-b)G(y)&=& \beta.  \\
\end{aligned}
\right.
\end{equation*}
We have $$y=g^{-1}(\alpha).$$
Moreover,
\begin{equation*}
\label{GGFeq}
z=\beta + (b-\gamma)G\left(g^{-1}(\alpha)\right).
\end{equation*}
Hence, we obtain
$$\psi^{-1}\left(\alpha,\beta \right)=\left( g^{-1}(\alpha) ,    \beta + (b-\gamma)G\left(g^{-1}(\alpha)\right)  \right).$$

Finally, we compute the compositional inverse of $ f(x) $.
From Theorem \ref{main}, together with the compositional inverses of $\phi, \psi$ and $\alpha = \lambda(x)$, $\beta=x-\lambda(x)$, the compositional inverse of $ f(x) $ is
\begin{equation*}
\begin{aligned}
f^{-1}(x)&= {\phi}^{-1} \circ \psi^{-1} \circ {\phi}(x)     \\
&=   g^{-1}(\alpha) + \beta + (b-\gamma)G\left(g^{-1}(\alpha)\right)    \\
&=   (b-\gamma)G\left(g^{-1}(\lambda(x))\right)+g^{-1}(\lambda(x))  - \lambda(x) +x .
\end{aligned}
\end{equation*}
\end{proof}

Theorem \ref{mainGGF} can be verified drectly by $f^{-1}(f(x)) =x$ for any $x\in \gf_q$.
		Firstly, for any $x\in\gf_{q}$, we have
		\begin{eqnarray*}
			g^{-1}(\lambda(f(x))) &=& g^{-1}\left( \lambda (x+\gamma G(\lambda (x))) \right) \\
			&=& g^{-1}(\lambda(x)+bG(\lambda(x))) = g^{-1}(g(\lambda(x))) = \lambda(x),
		\end{eqnarray*}
		where the second equality is due to the fact that $ \gamma $ is a $b$-linear translator with respect to $S$ for the map $\lambda$.  Therefore we obtain
		\begin{eqnarray*}
			f^{-1}(f(x)) &=& (b-\gamma) G(\lambda(x)) + \lambda(x) -\lambda (x+\gamma G(\lambda (x))) +x + \gamma G(\lambda (x)) \\
			&=& (b-\gamma) G(\lambda(x)) + \lambda(x) - (\lambda(x)+bG(\lambda(x))) + x  + \gamma G(\lambda (x)) = x.
		\end{eqnarray*}

Note that \cite[Theorem 1.2]{tuxanidy2014inverses} requires its $ \psi $ to be additive.
However, we focus more on $ \lambda(x) $ such that it has a $b$-linear translator with respect to $ S $, especially the case when $ \lambda(x) $ is not {additive} and even $ \lambda(0) \ne 0 $.
Specifically, when $\lambda(x)= \tr_{q^{n}/q}(x) $ or $ G $ is a $ q $-polynomial, Theorem \ref{mainGGF} is consistent with \cite[Corollary 1.6]{tuxanidy2014inverses}.

A special case of Theorem \ref{mainGGF} with $ G(x) = x$ (see \cite[Corollary 6.5]{akbary2011constructing}) is the following corollary:
\begin{Cor}
\cite[Theorem 3]{kyureghyan2011constructing}
\label{mainGGFc}
If $ b \ne -1 $, then the compositional inverse of the permutation $ x + \gamma \lambda(x) $ on $ \gf_q $ is $ \frac{-\gamma}{b+1}\lambda(x) +x $.
\end{Cor}

Furthermore, we propose a criterion for PPs of the form $x + \gamma G(\lambda(x)) $ being involutions,  and give some involutory constructions.
\begin{Cor}
\label{cominvolution2}
{
	Suppose $ \gamma \in \gf_q^* $.
	Then the} PP $f(x) = x + \gamma G(\lambda(x))  $ over $ \gf_{q} $ defined as in Lemma \ref{2-6.4} is an involution if and only if
{
\begin{enumerate}[(1)]
\item $ bG(y) + bG(y + bG(y))=0 $, for $ y \in \lambda(\gf_{q}) $ and   \label{cond1}
\item {anyone of the following holds:
\begin{enumerate}
\item[(i)] $b\neq 0$;
\item[(ii)] $q$ is even;
\item[(iii)]  $ G(\lambda(x)) = 0$ when $q$ is odd and $b = 0$.
\end{enumerate}}  \label{cond2}
\end{enumerate}
}
\end{Cor}
\begin{proof}
{Assume $f$ is an involution.  By Lemma \ref{maininvolution}}, we have $ bG(y) + bG(y + bG(y))=0 $ for $ y \in \lambda(\gf_{q}) $, where $  g(x)= x + bG(x) $ is defined as in Lemma \ref{2-6.4}.
{From now on, we assume Condition (1) holds.  Then, it suffices to prove that $f$ is an involution if and only if $\varphi(x)= 0 $.}
{Let $ \varphi(x) =  \gamma G(\lambda(x)) +\gamma G\left(  \lambda(x) + bG(\lambda(x))   \right)  $.}
Plugging $ bG(y) + bG(y + bG(y))=0 $ for $ y \in \lambda(\gf_{q}) $, $ g^{-1}(x)= g(x) $ and $ f^{-1}(x)$ by Theorem {\ref{mainGGF}} into $ \varphi(x) $, we have
$ f(x) $ is an involution if and only if $\varphi(x)=0$ holds for any $ x \in \gf_q $.
{
If $b\neq 0$, then Condition (\ref{cond1}) implies that $G(\lambda(x))=-G(\lambda(x)+b G(\lambda(x))) $
and thus $\varphi(x)= 0$ for any $ x \in \gf_q $.  If $b=0$ and $q$ is even,  then  $\varphi(x)=2 \gamma G(\lambda(x))=0$ as well.
If  $b=0$ and $q$ is odd, then
$$\varphi(x)=\gamma G(\lambda(x))+\gamma G(\lambda(x)+b G(\lambda(x)))=2 \gamma G(\lambda(x)).$$
In this case,   $\varphi(x) =0$ if and only if $G(\lambda(x)) =0$.

}
\end{proof}

When we consider $0$-linear translator in a finite field of even characteristic, the following explicit involution can be obtained easily.

\begin{Cor}
	\label{cominvolution3}
Let $ q $ be a power of $ 2 $.
Assume $ S \subseteq \gf_q $ and $ \lambda : \gf_q \rightarrow S $ is a surjective map.
Let $ \gamma \in \gf_q $ be a $0$-linear translator with respect to $S$ for the map $\lambda$.
Then for any $ G \in \gf_q[x] $ which maps $ S $ into $ S $, we have that $ f(x)= x + \gamma G(\lambda(x)) $ is an involution on $ \gf_q $.
\end{Cor}

{
We provide a specific construction as an example of Corollary \ref{cominvolution3}.
\begin{example}
Let  $ q $ be a power of $ 2 $ and $n>2$ be any integer.
Let $S=\gf_q$ and $ \lambda(x) = \sum\limits_{1 \le i < j \le n}{  \beta_i( x^{q^i} + x^{q^j}   )   }  \in \fqn[x] $. 
Clearly, each $ \gamma \in \gf_{q} $ is a $0$-linear translator with respect to $\gf_q$ for the map $\lambda(x)$.
Then for any $ G \in \gf_{q^n}[x] $ which maps $ \gf_q $ into $ \gf_q $, we have that $ f(x)= x + \gamma G  \left( \sum\limits_{1 \le i < j \le n}{  \beta_i( x^{q^i} +x^{q^j}   )   }   \right)  $ is an involution on $ \gf_{q^n} $.
\end{example}
}

{
\section*{Acknowledgement}
We thank  the editor  Sudhir R Ghorpade  and anonymous referees for their helpful suggestions.
}

%


\end{document}